\documentclass{llncs}
\pdfoutput=1
\usepackage{graphicx}
\usepackage{cite}
\usepackage{url}
\usepackage{times}
\usepackage{mathptm}

\DeclareSymbolFont{AMSb}{U}{msb}{m}{n}
\DeclareSymbolFontAlphabet{\Bbb}{AMSb}

\makeatletter
\def\hb@xt@{\hbox to }
\makeatother

\let\oldendproof\endproof
\def\endproof{\qed\oldendproof}

\begin{document}

\title{Isometric Diamond Subgraphs} 

\author{David Eppstein}

\institute{Computer Science Department\\
Donald Bren School of Information \& Computer Sciences\\
University of California, Irvine\\
\email{eppstein@uci.edu}}

\maketitle   

\begin{abstract}
We describe polynomial time algorithms for determining whether an undirected graph may be embedded in a distance-preserving way into the hexagonal tiling of the plane, the diamond structure in three dimensions, or analogous structures in higher dimensions. The graphs that may be embedded in this way form an interesting subclass of the partial cubes.
\end{abstract}

\section{Introduction}

In graph drawing, one seeks a layout of a given graph that optimizes legibility as measured, e.g., by vertex separation, angular resolution, or area.  Nearly ideal drawings result from subgraphs of regular tilings of the plane by squares or hexagons: the angular resolution is bounded, vertices have uniform spacing, all edges have unit length, and the area is at most quadratic in the number of vertices. For \emph{induced subgraphs} of the square or hexagonal tiling, one can additionally determine the graph's edges from the vertex placement: two vertices are adjacent if and only if they are mutual nearest neighbors.

Unfortunately, grid drawings are hard to find: Bhatt and Cosmodakis~\cite{BhaCos-IPL-87} showed the NP-com\-plete\-ness of testing whether a given graph is a subgraph of a square tiling. Eades and Whitesides~\cite{EadWhi-TCS-96}, generalizing this and several related NP-completeness results, described the \emph{logic engine} proof technique whereby they proved that it is NP-hard to determine whether a graph has a realization as a nearest-neighbor graph in the plane, or as a planar graph with unit edge lengths. Although it does not seem to have been stated explicitly previously, the same logic engine technique proves in a straightforward way the NP-completeness of determining whether a given graph is a subgraph of the hexagonal tiling, or an induced subgraph of the square or hexagonal tilings.

We previously showed that a stronger constraint on the embedding alleviates the computational difficulty of finding it: a polynomial time algorithm can test whether a graph embeds \emph{isometrically} onto the square tiling, or onto any higher dimensional integer grid of fixed or variable dimension~\cite{Epp-EJC-05}. In an isometric embedding, the unweighted distance between any two vertices in the graph equals the $L_1$ distance of their placements in the grid. An isometric embedding must be an induced subgraph, but not all induced subgraphs are isometric. Isometric square grid embeddings may be directly used as graph drawings, while planar projections of higher dimensional embeddings can be used to draw any \emph{partial cube}~\cite{Epp-GD-04}, a class of graphs with many applications~\cite{EppFalOvc-08}.

Can we find similar embedding algorithms for other tilings or repeating patterns of vertex placements in the plane and space? In this paper, we describe a class of $d$-dimensional infinitely repeating patterns, the \emph{generalized diamond structures}, which include the tiling of the plane by regular hexagons and the three-dimensional molecular structure of the diamond crystal. As we show, we can recognize in polynomial time the finite graphs that can be embedded isometrically onto a generalized diamond of any fixed or variable dimension, whenever such an embedding exists, in polynomial time.

\section{Hexagons and diamonds from projected slices of lattices}

\begin{figure}[t]
\centering
\includegraphics[width=1.7in]{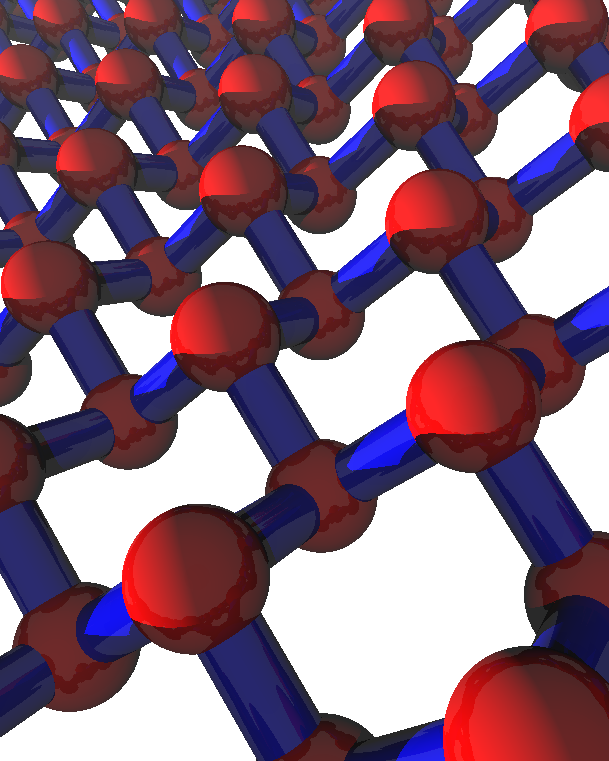}\qquad
\includegraphics[width=2.4in]{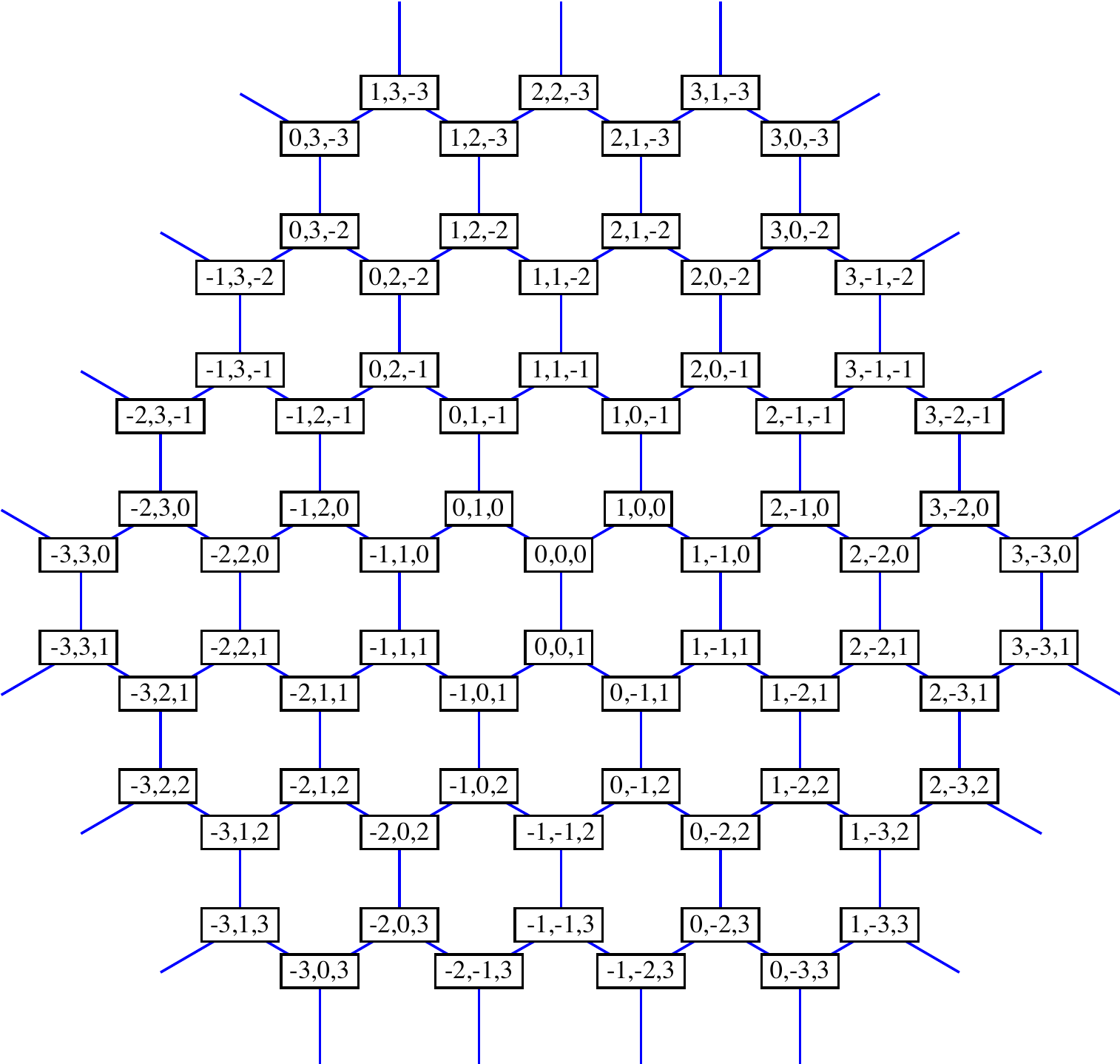}
\caption{Left: The unit distance graph formed by the integer points $\{(x,y,z)\mid x+y+z\in\{0,1\}\}$.
Right: The same graph projected onto the plane $x+y+z=0$ to form a hexagonal tiling.}
\label{fig:hexpov}
\end{figure}

\begin{figure}[t]
\centering
\includegraphics[height=2.2in]{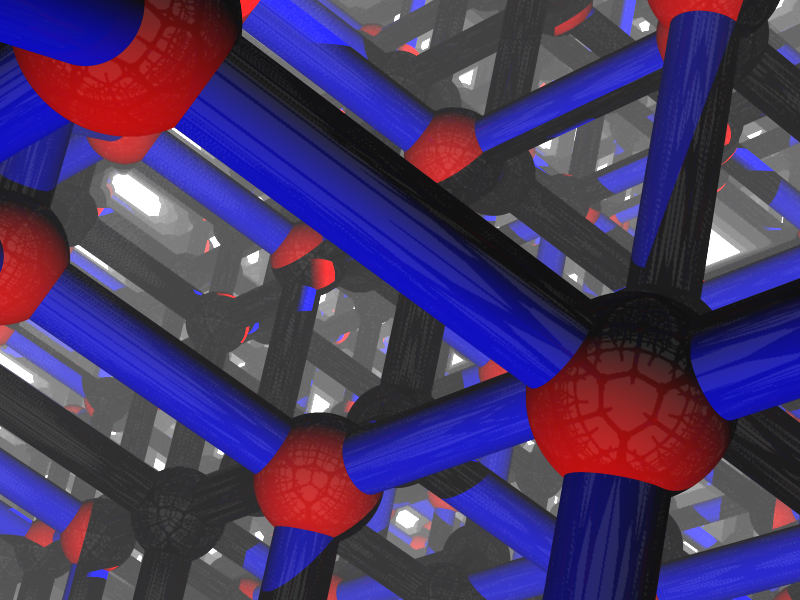}
\caption{The three-dimensional diamond graph.}
\label{fig:diapov}
\end{figure}

The three-dimensional points $\{(x,y,z)\mid x+y+z\in\{0,1\}\}$, with edges connecting pairs of points one unit apart, form a 3-regular infinite graph (Figure~\ref{fig:hexpov}, left) in which every axis-parallel line through any vertex passes through exactly one other vertex, forming a grid embedding of the graph~\cite{Epp-GD-08}.
The projection onto the plane $x+y+z=0$ gives a planar embedding in which each face is a regular hexagon (Figure~\ref{fig:hexpov}, right). We may repeat the same exercise in one higher dimension: the four-dimensional points $\{(w,x,y,z)\mid w+x+y+z\in\{0,1\}\}$, with edges connecting pairs of points one unit apart, projected into the three-dimensional subspace $w+x+y+z=0$, form an infinite 4-regular graph embedded in space in such a way that all edges have equal length and the four edges at any vertex spread in the pattern of a regular tetrahedron (Figure~\ref{fig:diapov}). This pattern of point placements and edges is realized physically by the crystal structure of diamonds, and so is often called the \emph{diamond lattice}, although it is not a lattice in the mathematical definition of the word; we call it the \emph{diamond graph}.

Analogously, we define a $k$-dimensional \emph{generalized diamond graph} in any dimension as follows. From the set of $(k+1)$-dimensional integer points such that the sum of coordinates is either zero or one, form a $(k+1)$-regular graph connecting pairs of points at unit distance. Each edge of this graph lies along one of the coordinate axes, and the edges are symmetrically arranged around any vertex. Project this graph orthogonally onto the hyperplane in which the coordinate sum of any point is zero; this projection preserves the symmetries of the original graph and produces a highly symmetric infinite $(k+1)$-regular graph embedded in $k$-dimensional space. Every point in this graph may be labeled with its $(k+1)$-dimensional coordinates, integers that sum to zero or one.

The generalized diamond graph is an isometric subset of the $(k+1)$-dimensional integer lattice: any two points may be connected by a path that has length equal to their $L_1$ distance. Thus, any finite isometric subgraph of the generalized diamond graph is a partial cube. However, not every partial cube may be embedded into the generalized diamond graphs; for instance, a square, cube, or hypercube itself may not be so embedded, because these graphs contain four-vertex cycles while the generalized diamond graphs do not. Thus we are led to the questions of which graphs are isometric diamond subgraphs, and how efficiently we may recognize them.

\section{Coherent cuts}

A \emph{cut} in a graph is a partition of the vertices into two subsets $C$ and $V\setminus C$; an edge \emph{spans} the cut if it has one endpoint in $C$ and one endpoint in $V\setminus C$. If $G=(U,V,E)$ is a bipartite graph, we say that a cut $(C,(U\cup V)\setminus C)$ is \emph{coherent} if, for every edge $(u,v)$ that spans the cut (with $u\in U$ and $v\in V$), $u$ belongs to $C$ and $v$ belongs to $(U\cup V)\setminus C$. That is, if we color the vertices black and white, all black endpoints of edges spanning the cut are on one side of the cut, and all white endpoints are on the other side.

The \emph{Djokovic--Winkler relation} of a partial cube $G$ determines an important family of cuts. Define a relation $\sim$ on edges of $G$ by $(p,q)\sim (r,s)$ if and only if $d(p,r)+d(p,s)=d(q,r)+d(q,s)$; then $G$ is a partial cube if and only if it is bipartite and $\sim$ is an equivalence relation~\cite{Djo-JCTB-73,Win-DAM-84}. Each equivalence class of $G$ spans a cut $(C,V\setminus C)$; we call $V$ and $V\setminus C$ \emph{semicubes}~\cite{Epp-EJC-05}. One may embed $G$ into a hypercube by choosing one coordinate per Djokovic--Winkler equivalence class, set to $0$ within $C$ and to $1$ within $V\setminus C$. Since this embedding is determined from the distances in $G$, the isometric embedding of $G$ into a hypercube is determined uniquely up to symmetries of the cube.

Figure~\ref{fig:coherent-desargues} depicts an example, the \emph{Desargues graph}. This is a symmetric graph on 20 vertices, the only known nonplanar 3-regular cubic partial cube~\cite{Epp-EJC-06}; it is used by chemists to model the configuration spaces of certain chemical compounds~\cite{Bal-RRC-66,Mis-ACR-70}. The left view is a more standard symmetrical view of the graph while the right view has been rearranged to show more clearly the cut formed by one of the Djokovic--Winkler equivalence classes. As can be seen in the figure, this cut is coherent: each edge spanning the cut has a blue endpoint in the top semicube and a red endpoint in the bottom semicube. The Djokovic--Winkler relation partitions the edges of the Desargues graph into five equivalence classes, each forming a coherent cut.

\begin{figure}[t]
\centering\includegraphics[width=3in]{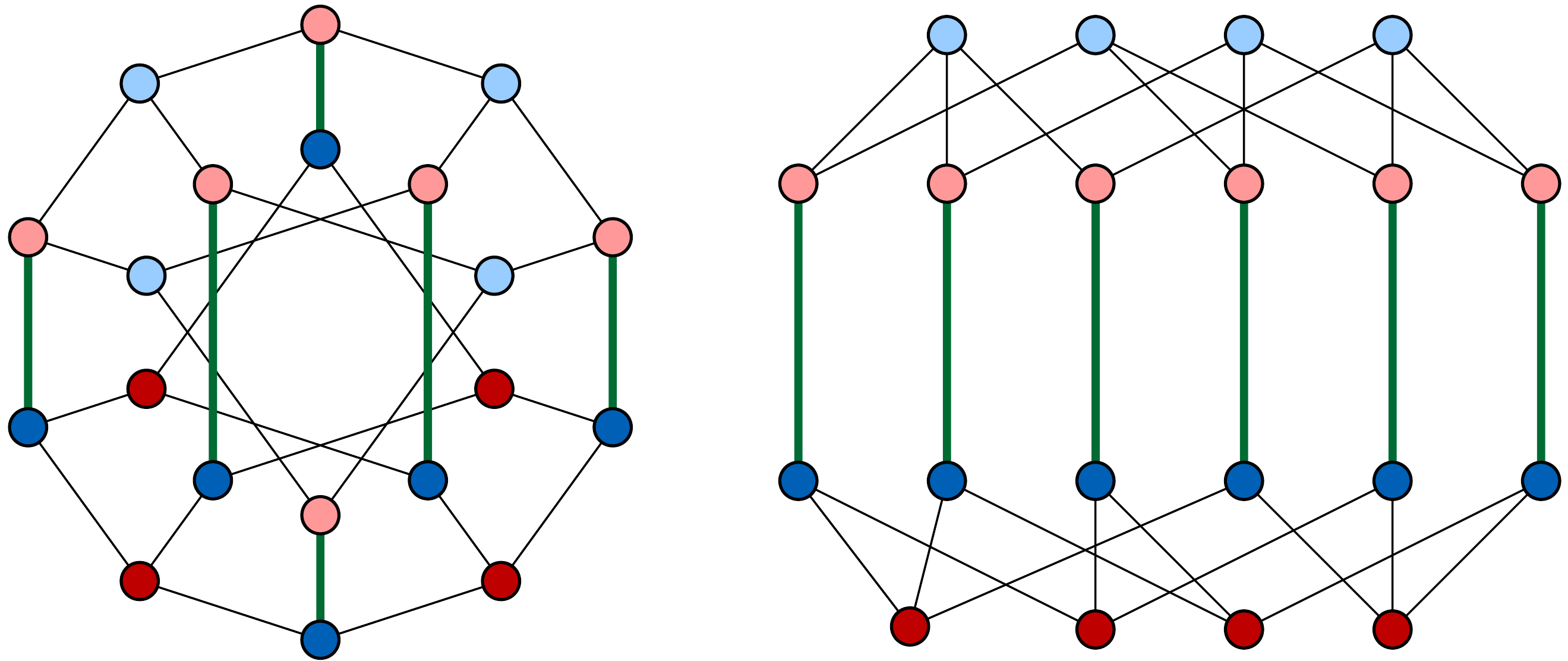}
\caption{Two views of the Desargues graph and of a coherent cut formed by a Djokovic--Winkler equivalence class.}
\label{fig:coherent-desargues}
\end{figure}

\begin{theorem}
\label{x-gdg}
A partial cube is an isometric subgraph of a generalized diamond graph if and only if all cuts formed by Djokovic--Winkler equivalence classes are coherent.
\end{theorem}

\begin{proof}
In the generalized diamond graph itself, each semicube consists of the set of points in which some coordinate value is above or below some threshold, and each edge spanning a Djokovic--Winkler cut connects a vertex below the threshold to one above it.  The bipartition of the generalized diamond graph consists of one subset of vertices for which the coordinate sum is zero and another for which the coordinate sum is one.  In an edge spanning the cut, the endpoint on the semicube below the threshhold must have coordinate sum zero and the other endpoint must have coordinate sum one, so the cut is coherent. The Djokovic--Winkler relation in any isometric subset of a generalized diamond graph is the restriction of the same relation in the generalized diamond itself, and so any isometric diamond subgraph inherits the same coherence property.

Conversely let $G$ be a partial cube in which all Djokovic--Winkler cuts are coherent; color $G$ black and white. Choose arbitrarily some white base vertex $v$ of $G$ to place at the origin of a $d$-dimensional grid, where $d$ is the number of Djokovic--Winkler equivalence classes, and assign a different coordinate to each equivalence class, where the $i$th coordinate value for a vertex $w$ is zero if $v$ and $w$ belong to the same semicube of the $i$th equivalence class, $+1$ if $v$ belongs to the white side and $w$ belongs to the black side of the $i$th cut, and $-1$ if $v$ belongs to the black side and $w$ belongs to the white side of the cut. This is an instance of the standard embedding of a partial cube into a hypercube by its Djokovic--Winkler relationship, and (by induction on the distance from $v$) every vertex has coordinate sum either zero or one. Thus, we have embedded $G$ isometrically into a $d$-dimensional generalized diamond graph.
\end{proof}

For example, the Desargues graph is an isometric subgraph of a five-dimensional generalized diamond.

\section{The diamond dimension}

Theorem~\ref{x-gdg} leads to an algorithm for embedding any isometric diamond subgraph into a generalized diamond graph, but possibly of unnecessarily high dimension. Following our previous work on \emph{lattice dimension}, the minimum dimension of an integer lattice into which a partial cube may be isometrically embedded~\cite{Epp-EJC-05}, we define the \emph{diamond dimension} of a graph $G$ to be the minimum dimension of a generalized diamond graph into which $G$ may be isometrically embedded. The diamond dimension may be as low as the lattice dimension, or (e.g., in the case of a path) as large as twice the lattice dimension.
We may compute the diamond dimension in polynomial time, as we now show. The technique is similar to that for lattice dimension, but becomes somewhat simpler in the generalized diamond case.

Color the graph black and white, and let $(C,V\setminus C)$ and $(C',V\setminus C')$ be two cuts determined by equivalence classes of the Djokovic--Winkler relation, where $C$ and $C'$ contain the white endpoints of the edges spanning the cut and the complementary sets contain the black endpoints. Partially order these cuts by the set inclusion relationship on the sets $C$ and $C'$:  $(C,V\setminus C) \le (C',V\setminus C')$ if and only if $C\subseteq C'$. The choice of which coloring of the graph to use affects this partial order only by reversing it.
A \emph{chain} in a partial order is a set of mutually related elements, an \emph{antichain} is a set of mutually unrelated elements, and the \emph{width} of the partial order is the maximum size of an antichain. By Dilworth's theorem~\cite{Dil-AM-50} the width is also the minimum number of chains into which the elements may be partitioned. Computing the width of a given partial order may be performed by transforming the problem into graph matching, but even more efficient algorithms are possible, taking time quadratic in the number of ordered elements and linear in the width~\cite{FelRagSpi-Ord-03}.

\begin{theorem}
\label{thm:dd}
The diamond dimension of any isometric diamond subgraph $G$, plus one, equals the width of the partial order on Djokovic--Winkler cuts defined above.
\end{theorem}

\begin{proof}
First, the diamond dimension plus one is greater than or equal to the width of the partial order. For, suppose that $G$ is embedded as an isometric subgraph of a $d$-dimensional generalized diamond graph; recall that this graph may itself be embedded isometrically into a $(d+1)$-dimensional grid. We may partition the partial order on cuts into $d+1$ chains, by forming one chain for the cuts corresponding to sets of edges parallel to each of the $d+1$ coordinate axes. The optimal chain decomposition of the partial order can only use at most as many chains.

In the other direction, suppose that we have partitioned the partial order on cuts into some family of $d+1$ chains. We may use this partition to embed $G$ isometrically into a $d$-dimensional generalized diamond graph: we let each chain correspond to one dimension of a $(d+1)$-dimensional integer lattice, place an arbitrarily-chosen white vertex at the origin, and determine the coordinates of each vertex by letting traversal of an edge in the direction from white to black increase the corresponding lattice coordinate by one unit. Each other vertex is connected to the origin by a path that either has equal numbers of white-to-black and black-to-white edges (hence a coordinate sum of zero) or one more white-to-black than black-to-white edge (hence a coordinate sum of one). Thus, the diamond dimension of $G$ is at most $d$.
As we have upper bounded and lower bounded the diamond dimension plus one by the width, it must equal the width.
\end{proof}

Thus, we may test whether a graph may be embedded into a generalized diamond graph of a given dimension, find the minimum dimension into which it may be so embedded, and construct an embedding of minimum dimension, all in polynomial time. To do so, find a partial cube representation of the graph, giving the set of Djokovic--Winkler cuts~\cite{Epp-SODA-08}, form the partial order on the cuts, compute an optimal chain decomposition of this partial order~\cite{FelRagSpi-Ord-03}, and use the chain decomposition to form an embedding as described in the proof of the theorem.
It would be of interest to find more general algorithms for testing whether a graph may be isometrically embedded into any periodic tiling of the plane, or at least any periodic tiling that forms an infinite partial cube. Currently, the only such tilings for which we have such a result are the square tiling~\cite{Epp-EJC-05} and (by the dimension two case of Theorem~\ref{thm:dd}) the hexagonal tiling.

\raggedright
\bibliographystyle{abbrv}
\bibliography{media}

\end{document}